\documentclass[11pt,reqno]{amsart}

\usepackage{amssymb,amsmath,amsthm,amscd,latexsym,amsfonts}
\usepackage{mathtools}
\usepackage[T1]{fontenc}
\usepackage{graphicx,epstopdf}
\usepackage{graphicx}
\usepackage{xcolor}
\usepackage{cite}
\usepackage{MnSymbol}
\usepackage[a4paper,top=3cm, bottom=3cm, left=3.1cm, right=3.1cm]{geometry}

\newtheorem{thm}{Theorem}
\newtheorem{defn}{Definition}
\newtheorem{lemma}{Lemma}
\newtheorem{pro}{Proposition}
\newtheorem{rk}{Remark}
\newtheorem{cor}{Corollary}

\numberwithin{equation}{section} 

\newcommand{\bea}{\begin{eqnarray}}
	\newcommand{\eea}{\end{eqnarray}}

\newcommand{\Z}{\mathbb{Z}}

\def\Z{\mathbb{Z}}

\begin{document}
	\title [Gibbs measures for HC-model ]
	{Gibbs measures for HC-model with a countable set of spin values on a Cayley tree}
	
	\author {R.M. Khakimov, M.T. Makhammadaliev,  U.A. Rozikov}
	
	\
	\address{R.M. Khakimov, M.T. Makhammadaliev$^{a,b}$
		\begin{itemize}
			\item[$^a$] V.I.Romanovskiy Institute of Mathematics,  9, Universitet str., 100174, Tashkent, Uzbekistan;
			\item[$^b$] Namangan State  University, Namangan, Uzbekistan.
	\end{itemize}}
	\email{rustam-7102@rambler.ru, mmtmuxtor93@mail.ru}
	
	\address{ U.Rozikov$^{a,c,d}$
		\begin{itemize}
			\item[$^a$] V.I.Romanovskiy Institute of Mathematics,  9, Universitet str., 100174, Tashkent, Uzbekistan;
			\item[$^c$] AKFA University, National Park Street, Barkamol MFY,
			Mirzo-Ulugbek district, Tashkent, Uzbekistan;
			\item[$^d$] National University of Uzbekistan,  4, Universitet str., 100174, Tashkent, Uzbekistan.
	\end{itemize}}
	\email{rozikovu@yandex.ru}
	
	\begin{abstract}
	In this paper, we study the HC-model with a countable set $\mathbb Z$ of spin values on a Cayley tree
	 of order $k\geq 2$. This model is defined
by	a countable set of parameters (that is, the activity function $\lambda_i>0$, $i\in \mathbb Z$).
	A functional equation is obtained that provides the consistency condition for 
	finite-dimensional Gibbs distributions. Analyzing this equation,
	the following results are obtained:
	\begin{itemize}
		\item[-] Let $\Lambda=\sum_i\lambda_i$.
		For $\Lambda=+\infty$ there are no translation-invariant Gibbs measures (TIGM) and no two-periodic Gibbs measures (TPGM);
		\item[-] For $\Lambda<+\infty$, the uniqueness of TIGM is proved;
		\item[-] Let $\Lambda_{\rm cr}(k)=\frac{k^k}{(k-1)^{k+1}}$. If $0<\Lambda\leq\Lambda_{\rm cr}$, then there is exactly one TPGM that is TIGM;
		\item[-] For $\Lambda>\Lambda_{\rm cr}$, there are exactly three TPGMs, one of which is TIGM.
	\end{itemize}
		
	\end{abstract}
	\maketitle
	
	
	
	\section{Introduction}

 The theory of Gibbs measures is well developed in many classical models from physics (for example, the Ising model, the Potts model, the HC model), when the set of spin values is a finite set. It is known that each Gibbs measure corresponds to one phase of the physical system.
Therefore, in the theory of Gibbs measures, one of the important problems is the existence
and non-uniqueness of Gibbs measures. The non-uniqueness means that the physical system has a coexistence of its phases (state)
at a fixed temperature (see  \cite{FV},  \cite{6}, \cite{Pr}, \cite{R}, \cite{Si}).

In the case of models with a finite set of spin values, the set of all limit Gibbs measures for a given Hamiltonian forms a non-empty, convex, compact subset in the set of all probability measures. In the setting of gradient interface models
in general no Gibbs measure exists \cite{Ve}. Therefore one often
considers gradient Gibbs measures \cite{FS}.

There are papers devoted to the study of (gradient) Gibbs measures for models with an infinite set of spin values. In particular, \cite{GR} shows the uniqueness of the translation-invariant Gibbs measure for the antiferromagnetic Potts model with a countable number of states and a nonzero external field, and \cite{G} describes the Poisson measures, which are Gibbs measures. 

The work \cite{Z} is devoted to the study of Gibbs measures for models of gradient type. 

In \cite{HKR} the existence of several transnational invariant gradient Gibbs measures for the SOS model with a countable set  of spin values on the Cayley tree is proven. And also the class of 4-periodic Gibbs gradient measures is described. Recently, in \cite{HK} (gradient) Gibbs measures for a gradient-type model are studied and the existence of a countable set of ordinary Gibbs measures is shown, and conditions for the existence of gradient Gibbs measures that are not ordinary Gibbs measures are found. In \cite{B} the existence of a relationship between Gibbs measures and Gibbs gradient measures is shown. In \cite{HK1}  the existence of gradient Gibbs measures that are not translation-invariant is proved.

In this paper, we study HC-model with a countable set of spin values. Such HC-models are interesting in statistical mechanics, combinatorics and theory of 
neural networks \cite{bhw}, \cite {6},  \cite{Kf}, \cite{Maz}. Many papers are devoted to the study of limit Gibbs measures for HC-models with a finite set of spin values (see, for example, \cite{RM}, \cite{R} and the references therein).

Here for HC-model with a countable set of spin values $\mathbb Z$ on a Cayley tree of arbitrary order
we will find some conditions for the existence of TIMG,
and also prove the uniqueness of such a measure under the existence condition. Besides,
for the model under consideration, two-periodic Gibbs measures  are studied.

The exact value $\Lambda_{\rm cr}$  of the parameter $\Lambda$ is found, (where $\Lambda$ is the sum of the series obtained
from the sequence of parameters $\{\lambda_j\}_{j\in \mathbb Z}$), such that for $0<\Lambda\leq\Lambda_{\rm cr}$
there is exactly one periodic Gibbs measure which is translation invariant,
and for $\Lambda>\Lambda_{\rm cr}$ there are exactly three periodic Gibbs measures, one of which is translation invariant.

\section{Preliminaries}

The Cayley tree $\Gamma^k$
of order $ k\geq 1 $ is an infinite tree,
i.e., a graph without cycles, such that exactly $k+1$ edges
originate from each vertex. Let $\Gamma^k=(V,L,i)$, where $V$ is the
set of vertices $\Gamma^k$, $L$ is the set of edges and $i$ is the
incidence function setting each edge $l\in L$ into correspondence
with its endpoints $x, y \in V$. If $i (l) = \{ x, y \} $, then
the vertices $x$ and $y$ are called the {\it nearest neighbors},
denoted by $l = \langle x, y \rangle $.

For a fixed point $x^0\in V$,
$$W_n=\{x\in V\,| \, d(x,x^0)=n\}, \qquad V_n=\bigcup_{m=0}^n W_m, \qquad L_n=\{\langle x,y\rangle\in L| \, x,y\in V_n\},$$
where $d(x,y)$ is the distance between
vertices $x$ and $y$ on a Cayley tree, i.e.,
the number of edges of the shortest path connecting  $x$ and $y$.

Write  $x\prec y$, if the path from  $x^0$ to $y$ goes through $x$.
Call vertex $y$ a direct successor of $x$ if $y\succ x$ and $x,y$
are nearest neighbors. Note that in $\Gamma^k$ any vertex $x\neq x^0$
has $k$ direct successors and $x^0$ has $k+1$ direct successors. Denote by
$S(x)$ the set of direct successors of $x$, i.e. if $x\in W_n$, then
$$S(x)=\{y_i\in W_{n+1} \, |\,  d(x,y_i)=1, i=1,2,\ldots, k \}.$$

 We consider the Hard-Core (HC) model with a countable set of spin values in which the spin variables take values in the set $\mathbb{Z}$, and are located at the tree vertices. A configuration $\sigma = \{\sigma(x) \, |\, x \in  V \}$  is then defined as a function $\sigma=\{\sigma(x)\in \mathbb Z: x\in V\}$. In this model, each vertex $x$ is assigned one of the values $\sigma (x)\in \mathbb Z$, where $\mathbb{Z}$ is the set of integers. The values $\sigma (x)\neq0$ mean that the vertex $x$ is `occupied', and $\sigma (x)=0$ means that $x$ is `vacant'.

We consider the set $\mathbb{Z}$ as the set of vertices of a graph $G$.
We use the graph $G$  to define a $G$-admissible configuration as follows.
A configuration $\sigma$ is called a
$G$-\textit{admissible configuration} on the Cayley tree (in $V_n$), if $\{\sigma (x), \, \sigma (y)\}$ is the edge of the graph $G$
for any pair of nearest neighbors $x,y$ in $V$ (in $V_n$). We
let $\Omega^G$ ($\Omega_{V_n}^G$) denote the set of $G$-admissible configurations.

The activity set \cite{bw} for a graph $G$ is the bounded function $\lambda:G
\to R_+$ from the set $G$  to the set of positive real
numbers. The value $\lambda_i$ of the function $\lambda$ at the vertex
$i \in \mathbb{Z}$ is called the vertex activity.

For given $G$ and $\lambda$ we define the Hamiltonian of the $G-$HC model as
\begin{equation}\label{H} H^{\lambda}_{G}(\sigma)=\left\{%
\begin{array}{ll}
    J \sum\limits_{x\in{V}}{\ln\lambda_{\sigma(x)},} \ \ \ $ if $ \sigma \in\Omega^G $,$ \\
   +\infty ,\  \  \  \ \ \ \ \ \ \ \ \ \ $  \ if $ \sigma \ \notin \Omega^G $,$ \\
\end{array}
\right.
\end{equation}
where $J\in \mathbb R$.

%
%


For nearest-neighboring interaction  potential $\Phi=(\Phi_b)_b$, where
$b=\langle x,y \rangle$ is an edge,  define symmetric transfer matrices $Q_b$ by
\begin{equation}\label{Qd}
	Q_b(\omega_b) = e^{- \big(\Phi_b(\omega_b) + | \partial x|^{-1} \Phi_{\{x\}}(\omega(x)) + |\partial y |^{-1} \Phi_{\{y\}} (\omega(y)) \big)},
\end{equation}
where $\partial x$ is the set of all nearest-neighbors of $x$ and $|S|$ denotes the number of elements of the set $S$.

%
%

Define the Markov (Gibbsian) specification as
$$
\gamma_\Lambda^\Phi(\sigma_\Lambda = \omega_\Lambda | \omega) = (Z_\Lambda^\Phi)(\omega)^{-1} \prod_{b \cap \Lambda \neq \emptyset} Q_b(\omega_b).
$$


Let $L(G)$ be the set of edges of a graph $G$. We let $A\equiv A^G=\big(a_{ij}\big)_{i,j=0,1,2}$ denote the adjacency
matrix of the graph $G$, i.e.,

$$a_{ij}=a_{ij}^G=%
\begin{cases} 1 \ \ \mbox{if} \ \ \{i,j\}\in L(G), \\
0 \ \ \mbox{if} \ \ \{i,j\}\notin L(G).
\end{cases}
$$

\begin{rk} Since ${i,j}\in \mathbb{Z}$, then in cases where $j$ is a specific negative number, instead of $a_{ij}$ we will conventionally write $a_{ i,j}$.
\end{rk}
%
%

\begin{defn} (See \cite[Chapter 12]{6}, \cite{HK}) A family of vectors $l= \{l_{xy}\}_{\langle x, y \rangle \in L}$ with
$l_{xy}=\{l_{xy}(i):i \in \mathbb Z \} \in (0, \infty)^\mathbb Z$ is called the boundary law for the Hamiltonian (\ref{H}) if

1) for each $\langle x, y \rangle \in L $ there exists  a constant $c_{xy}> 0 $ such that the consistency equation
\begin{equation}\label{eq:bl}
l_{xy}(i) = c_{xy}\lambda_i\prod_{z \in \partial x \setminus \{y \}} \sum_{j \in \Z} a_{ij} l_{zx}(j)
\end{equation}
holds for every $ i \in \Z $, where $\partial x-$ the set of nearest neighbors of a vertex $x$.

2)  A boundary law $l$ is said to be {\em normalisable} if and only if
\begin{equation}\label{Norm}
\sum_{i \in \Z} \Big(\lambda_i \prod_{z \in \partial x} \sum_{j \in \Z} a_{ij} l_{zx}(j) \Big) < \infty
\end{equation} at any $x \in V$.

3) A boundary law 	is called {\em $q$-height-periodic} (or $q$-periodic) if $l_{xy} (i + q) = l_{xy}(i)$
	for every oriented edge $\langle x,y \rangle $ and each $i \in \Z$.
	
4) 	 A boundary law is called {\em translation invariant}  if it does not depend on edges of the tree, i.e., $l_{xy} (i) = l(i)$
for every oriented edge $\langle x,y \rangle $ and each $i \in \Z$.
\end{defn}
Assuming $l_{xy}(0)\equiv 1$ (normalization at $0$), from (\ref{eq:bl}) we get
\begin{equation}\label{ma}
l_{xy}(i) = {\lambda_i\over \lambda_0}\prod_{z \in \partial x \setminus \{y \}}{a_{i0}+ \sum_{j \in \Z_0} a_{ij} l_{zx}(j)\over a_{00}+\sum_{j \in \Z_0} a_{0j} l_{zx}(j)}.
\end{equation}

\begin{rk} \label{r2} We note that
	\begin{itemize}
\item[a.] There is an one-to-one correspondence between boundary laws
and Gibbs measures (i.e., tree-indexed Markov chains) if the boundary laws are  normalisable \cite{Z1} (see \cite[Theorem 3.5.]{HKR}).

\item[b.] In \cite{HKR} it is shown that a translation invariant boundary law $z\in \mathbb R_+^\infty$ satisfies the condition of normalisability, if $z\in l^{\frac{k+1}{k}}$.
\end{itemize}
\end{rk}

%

In this paper we consider the nearest-neighboring interaction  potential $\Phi=(\Phi_b)_b$, which corresponds to the HC model  (\ref{H}) and will study Gibbs measures of this model. By Remark \ref{r2} each normalisable boundary law $l$ defines a Gibbs measure. In this paper our aim is to find $1$-height-periodic and two-periodic boundary laws for the HC model for a specially chosen graph $G$ (see below). We show that these boundary laws will be normalisable and therefore define Gibbs measures.

We consider the graph $G$ with $a_{i0}=1$ for any $i\in \mathbb{Z}$ and $a_{im}=0$ for any $i, m \in \mathbb{Z}_0$ (see Fig.\ref{fi}). The corresponding admissible configuration satisfies the equality $\sigma (x)\sigma (y)=0$ for any $\langle x,y \rangle $ from $V$, i.e., if the vertex $x$ has the spin value $\sigma(x)=0$, then on neighboring vertices we can put any value from $\mathbb{Z}$, and   if the vertex $x$ contains any value from $\mathbb{Z}_0$, then on neighboring vertices we put only zeros.

\begin{figure}[h]
\begin{center}
   \includegraphics[width=13cm]{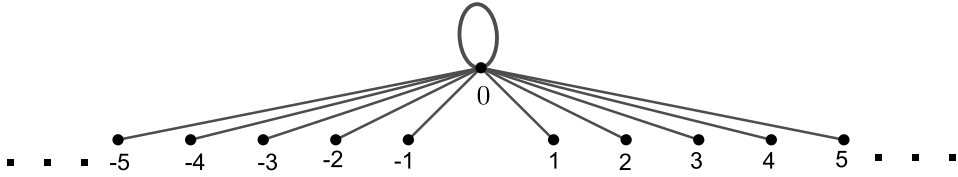}
\end{center}
     \caption{Countable graph $G$ with vertex set $\mathbb Z$, all vertices are connected only with $0$.}\label{fi}

\end{figure}

For this graph $G$, from (\ref{ma}) (see \cite{6} and \cite{BR}) by introducing new variables we obtain 
\begin{equation}\label{e10}
z_{i,x}=\lambda_i\prod_{y \in S(x)} {1\over 1+\sum_{j\in \mathbb{Z}_0} z_{j,y}}, \ \ i\in \mathbb{Z}_0. \end{equation}

\section{Translation invariant measures}

The problem of the finding of the general form of solutions of the equation (\ref{e10}) seems to be a very difficult.
In this subsection, we consider translation-invariant solutions, i.e., $z_x=z\in \mathbb R_{+}^{\infty}.$
In this case the equation (\ref{e10}) has the following form
\begin{equation}\label{e11}
z_i=\lambda_i\cdot \left({1\over 1+\sum_{j\in \mathbb Z_0} z_j}\right)^k, \ \ i\in \mathbb Z_0. \end{equation}
Here $\lambda_i>0, \ z_i>0$.

\begin{lemma}\label{L1} Let $k\geq2$. If there is a positive solution $\{z_j\}_{j\in \mathbb Z_0}$ of the system of equations (\ref{e11}) for some sequence of parameters $\{\lambda_j\}_{j\in \mathbb Z_0}$ then series $\sum_{j\in \mathbb Z_0} z_j$ and $\sum_{j\in \mathbb Z_0}\lambda_j$ obtained  respectively from $\{z_i\}_{i\in \mathbb Z_0}$ and $\{\lambda_j\}_{j\in \mathbb Z_0}$ converge.
\end{lemma}
\begin{proof} Let  $\{z_j\}_{j\in \mathbb Z_0}$ be a solution of the system of equations (\ref{e11}). We assume that the series $\sum_{j\in \mathbb Z_0} z_j$ diverges. Then, since $z_j>0$, it is obvious that $\sum_{j\in \mathbb Z_0} z_j=+\infty$. Hence due to (\ref{e11}) we get $z_i=0, \ i\in \mathbb Z_0$, i.e., $\sum_{i\in \mathbb Z_0} z_i<+\infty.$ This is a contradiction. So under the conditions of lemma the series $\sum_{j\in \mathbb Z_0} z_j$  converges.

Let  $\sum_{j\in \mathbb Z_0} z_j=A$. Then from (\ref{e11}) we obtain $z_i=\lambda_i\cdot \left({1\over 1+A}\right)^k$. Thence $\sum_{i\in \mathbb Z_0}\lambda_i=A(1+A)^k$, i.e., the series $\sum_{i\in \mathbb Z_0}\lambda_i$ converges. 
Lemma is proved.
\end{proof}

 By Lemma \ref{L1} it follows that there is no a positive solution of the system of equations (\ref{e11}) for which the series $\sum_{j\in \mathbb Z_0} z_j$ and $\sum_{j\in \mathbb Z_0}\lambda_j$ diverge, i.e., these conditions are necessary for the existence of a solution (\ref{e11}).

\begin{pro}\label{P1} Let $k\geq2$. If the series $\sum_{j\in \mathbb Z_0}\lambda_j$ obtained from a sequence of parameters $\{\lambda_j\}_{j\in \mathbb Z_0}$ converges then for the sequence $\{\lambda_j\}_{j\in \mathbb Z_0}$ there exists a unique positive solution $\{z_j\}_{j\in \mathbb Z_0}$ of the system of equations (\ref{e11}).
\end{pro}
\begin{proof} Let the series $\sum_{j\in \mathbb Z_0}\lambda_j$ converge and its sum is  $\sum_{j\in \mathbb Z_0}\lambda_j=\Lambda$. We will prove that for the sequence   $\{\lambda_j\}_{j\in \mathbb Z_0}$ there is a unique solution of the system of equations (\ref{e11}). By Lemma \ref{L1} it follows that for the existence of a solution $\{z_j\}_{j\in \mathbb Z_0}$ of the system of equations (\ref{e11}) the convergence of the series $\sum_{j\in \mathbb Z_0} z_j$ is necessary.

Let $\sum_{j\in \mathbb Z_0} z_j=A$. Then due to (\ref{e11}) we get $z_i=\frac{\lambda_i}{(1+A)^k}, \ i\in \mathbb Z_0.$ Hence
\begin{equation}\label{e12}
\sum_{j\in \mathbb Z_0}z_j=\frac{\sum_{j\in \mathbb Z_0}\lambda_j}{(1+A)^k},
\end{equation}
i.e.,
$$A(1+A)^k-\Lambda=0, \ A>0.$$
By the Descartes rule of signs, the last equation has only one positive root $A=A_0$. There are infinitely many sequences $\{z_j\}$ for which $\sum_{j\in \mathbb Z_0} z_j=A_0$. Among them the sequence $\{z_j\}$ for which there exists $\{\lambda_j\}$ such that $\{z_j\}$ satisfies (\ref{e11}) is unique. This follows from the equality $z_i=\frac{\lambda_i}{(1+A_0)^k}$ because ${\lambda_i}$ are fixed and $A_0$ is unique. 
\end{proof}
\begin{rk}  We note that the solution  $\{z_j\}_{j\in \mathbb Z_0}$ in Proposition \ref{P1} is normalisable because the convergence of series $\sum_{j\in \mathbb Z_0} z_j^{\frac{k+1}{k}}$ follows from the convergence of $\sum_{j\in \mathbb Z_0} z_j$. Then by Remark \ref{r2}  the Gibbs measure (denoted by $\mu_0$) corresponding to this solution exists.
\end{rk}

\textbf{Markov chain corresponding to the translation-invariant Gibbs measure.} 

For the Gibbs measure $\mu_0$ (of the unique  solution of the system of equations (\ref{e11})) we will check the existence of a stationary distribution of the Markov chain corresponding to the measure $\mu_0$. Consider the matrix $\mathbb P$ of transition probabilities  corresponding to the measure $\mu_0$:
$$\mathbb P =\begin{pmatrix}
 & \vdots & \vdots & \vdots & \vdots & \vdots &  \\
 \cdots & p_{-2-2} &  p_{-2-1}  & p_{-2 0} &  p_{-2 1} &  p_{-2 2} & \cdots \\
 \cdots & p_{-1-2} &  p_{-1-1}  & p_{-1 0} &  p_{-1 1} &  p_{-1 2} & \cdots \\
 \cdots & p_{0-2} &  p_{0-1}  & p_{0 0} &  p_{0 1} &  p_{0 2} & \cdots \\
 \cdots & p_{1-2} &  p_{1-1}  & p_{1 0} &  p_{1 1} &  p_{1 2} & \cdots \\
 \cdots & p_{2-2} &  p_{2-1}  & p_{2 0} &  p_{2 1} &  p_{2 2} & \cdots \\
 & \vdots & \vdots & \vdots & \vdots & \vdots & 
\end{pmatrix}.$$
Here
$$p_{\sigma(x)\sigma(y)}=\frac{a_{\sigma(x)\sigma(y)}\lambda_{\sigma{(y)}} z_{\sigma(y)}}{\sum_{\sigma(y) \in \mathbb Z}{a_{\sigma(x)\sigma(y)}\lambda_{\sigma{(y)}} z_{\sigma(y)}}} \ \  \Rightarrow \  \  p_{ij}=\frac{a_{ij}\lambda_{j}z_j}{\sum_{l \in \mathbb Z}{a_{il}\lambda_{l}z_l}}.$$

For the considered model, we have $\sigma(x)\sigma(y)=0$ for any nearest neighbors $\{x,y\}$. If $i\in \mathbb Z_0$ and $j \in \mathbb Z_0$ then $a_{ij}=0$, $a_{i0}=1$ and $a_{0j}=1$. Hence
$$p_{00}=\frac{\lambda_0z_0}{\lambda_0z_0+\sum_{l \in \mathbb Z_0}{\lambda_lz_l}},$$
$$ p_{01}=\frac{\lambda_1z_1}{\lambda_0z_0+\sum_{l \in \mathbb Z_0}{\lambda_lz_l}}, \ \ p_{0i}=\frac{\lambda_i z_i}{1+\sum_{l \in \mathbb Z_0}{\lambda_lz_l}}, $$
$$p_{10}=\frac{\lambda_1z_0}{\lambda_1z_0+\sum_{l \in \mathbb Z_0}{a_{1l}\lambda_lz_l}}=1, \ \
p_{i0}=\frac{\lambda_iz_0}{\lambda_iz_0+\sum_{l \in \mathbb Z_0}{a_{il}\lambda_lz_l}}=1.$$
 Therefore $\mathbb P$ has the following form (for $z_0=1$):
$$\mathbb P=\begin{pmatrix}
 & \vdots & \vdots & \vdots & \vdots & \vdots &  \\
 \cdots & 0 & 0 & 1  &  0 &  0 & \cdots \\
 \cdots & 0 & 0 & 1  &  0 &  0 & \cdots \\
 \cdots & \frac{\lambda_{-2} z_{-2}}{1+\sum_{l \in \mathbb Z_0}{\lambda_lz_l}} & \frac{\lambda_{-1} z_{-1}}{1+\sum_{l \in \mathbb Z_0}{\lambda_lz_l}} &  p_{00}  &  \frac{\lambda_1 z_1}{1+\sum_{l \in \mathbb Z_0}{\lambda_lz_l}}&  \frac{\lambda_2 z_2}{1+\sum_{l \in \mathbb Z_0}{\lambda_lz_l}} & \cdots \\
 \cdots & 0 & 0 & 1  &  0 &  0 & \cdots \\
 \cdots & 0 & 0 & 1  &  0 &  0 & \cdots \\
  & \vdots & \vdots & \vdots & \vdots & \vdots & 
 \end{pmatrix}.$$

We consider the vector $X=(\cdots,x_{-2},x_{-1},x_0,x_1,x_2,\cdots), \ \sum_{j\in \mathbb Z}{x_j}=1.$ If the system of equations $X\cdot\mathbb P=X$ has a solution then there exists a stationary distribution of the Markov chain corresponding to the measure $\mu_0$. So we solve the equation $X\cdot\mathbb P=X$. We have
$$X\cdot \mathbb P=\left(\cdots,\frac{x_0\lambda_{-1}z_{-1}}{1+\sum_{l \in \mathbb Z_0}{\lambda_lz_l}},1-\frac{x_0\sum_{l \in \mathbb Z_0}{\lambda_lz_l}}{1+\sum_{l \in \mathbb Z_0}{\lambda_lz_l}},\frac{x_0\lambda_1z_{1}}{1+\sum_{l \in \mathbb Z_0}{\lambda_lz_l}}, \frac{x_0\lambda_{2}z_{2}}{1+\sum_{l \in \mathbb Z_0}{\lambda_lz_l}},\cdots\right).$$
From the equality $X\cdot\mathbb P=X$ we get
\begin{equation}\label{e13}
x_0=1-\frac{x_0\sum_{l \in \mathbb Z_0}{\lambda_lz_l}}{1+\sum_{l \in \mathbb Z_0}{\lambda_lz_l}}, \  \ \ x_j=\frac{x_0\lambda_{j}z_{j}}{1+\sum_{l \in \mathbb Z_0}{\lambda_lz_l}}, \ \ \  j\in \mathbb Z_0.
\end{equation}
From the first equality of (\ref{e13}) we  find $x_0$:
$$x_0=\frac{1+\sum_{l \in \mathbb Z_0}\lambda_{l}z_{l}}{1+2\sum_{l \in \mathbb Z_0}\lambda_{l}z_{l}}.$$
Using the expression for $x_0$ from the second equality of (\ref{e13}) we find $x_j$:
$$x_j=\frac{\lambda_{j}z_{j}}{1+2\sum_{l \in \mathbb Z_0}\lambda_{l}z_{l}}, \ j\in \mathbb Z_0.$$
It is easy to see that the resulting vector is stochastic:
$$\sum_{j\in \mathbb Z}{x_j}=\frac{\sum_{l \in \mathbb Z_0}\lambda_{l}z_{l}}{1+2\sum_{l \in \mathbb Z_0}\lambda_{l}z_{l}}+
\frac{1+\sum_{l \in \mathbb Z_0}\lambda_{l}z_{l}}{1+2\sum_{l \in \mathbb Z_0}\lambda_{l}z_{l}}=1.$$
Hence there exists a stationary distribution of the Markov chain corresponding to the measure $\mu_0$.

Due to the uniqueness of the stationary distribution from Theorem 2 in \cite{Sh}, (p.612) we get

\begin{cor} In the set of states $\mathbb Z$ of a Markov chain with the transition probabilities matrix $\mathbb P$, there is exactly one positive recurrent class of essential
	communicating states (for definitions, see Chapter VIII in \cite{Sh}).
\end{cor}

Summarizing we the following 

\begin{thm} Let $k\geq2$. Then for the HC model with a countable set of states  (corresponding to the graph from Fig.\ref{fi}) the following statements are true:
\begin{itemize}
\item[1.] If the series $\sum_{j\in \mathbb Z_0}\lambda_j$ obtained from a sequence of parameters $\{\lambda_j\}_{j\in \mathbb Z_0}$ converges then there exists a unique translation-invariant Gibbs measure.

\item[2.] If the series $\sum_{j\in \mathbb Z_0}\lambda_j$  diverges there is no translation-invariant Gibbs measure.
\end{itemize}
\end{thm}
\textbf{Example 1.} Let $k=2$. Find  $\{z_j\}_{j\in \mathbb Z_0}$ satisfying the system of equations (\ref{e11}) for the sequence
 $$\lambda_j=\left\{%
\begin{array}{ll}
\frac{9}{4(4j-3)(4j-1)}, \ \ j\in \mathbb N, \\
\frac{9}{4(4j-1)(4j+1)}, \ \ j\in \mathbb {Z}_0\setminus \mathbb N.
\end{array}%
\right.$$

\textbf{Solution.} First we find the sum of the series $\sum_{j\in \mathbb {Z}_0}{\lambda_j}$.
 $$\Lambda=\sum_{j\in \mathbb {Z}_0}{\lambda_j}=\frac94\cdot\sum_{n=1}^\infty{\frac{1}{(2n-1)(2n+1)}}=\frac94\Big(\frac{1}{1\cdot3}+\frac{1}{3\cdot5}+\frac{1}{5\cdot7}+...\Big)=$$
$$=\frac98\Big(\frac11-\frac13\Big)+\frac98\Big(\frac13-\frac15\Big)+\frac98\Big(\frac15-\frac17\Big)+...+\frac98\Big(\frac{1}{2n-1}-\frac{1}{2n+1}\Big)+...=$$
$$=\frac98\cdot\Big(\frac11-\frac13+\frac13-\frac15+\frac15-\frac17+\frac17+...\Big)=\frac98.$$

For $k=2$ and $\Lambda=\frac98$ from the equality $\Lambda=A(1+A)^2$ we find $A_0=\frac12$. From the equality $z_i=\frac{\lambda_i}{(1+A_0)^2}$ we find terms of the series:
 $$z_1=\frac34\cdot\frac49=\frac13, \ z_2=\frac{9}{140}\cdot\frac49=\frac{1}{35}, \ \ldots, \ z_n=\frac{1}{(4n-3)(4n-1)}, \ \ldots,$$
$$z_{-1}=\frac{3}{20}\cdot\frac49=\frac{1}{15}, \ z_{-2}=\frac{1}{28}\cdot\frac49=\frac{1}{63}, \ \ldots, \ z_{-n}=\frac{1}{(4n-1)(4n+1)}, \ \ldots.$$
Hence the series $\sum_{j\in \mathbb Z_0}{z_j}$ has the following form:
$$\sum_{j\in \mathbb Z_0}{z_j}=\sum_{n=-1}^{-\infty}\frac{1}{(4n-1)(4n+1)}+\sum_{n=1}^\infty\frac{1}{(4n-3)(4n-1)}=\frac12.$$

\textbf{Example 2.} For $k=2$ and for a sequence of parameters given by the Poisson distribution
 $$\lambda_j=\left\{%
\begin{array}{ll}
\frac{2.4^j}{j!}\cdot e^{-2.4}, \ \ j\in \mathbb N, \\ [3mm]
\frac{8^{-j}}{\mid j\mid!}\cdot e^{-8}, \ \ j\in \mathbb {Z}_0\setminus \mathbb N,
\end{array}%
\right.$$
find a series $\sum_{j\in \mathbb Z_0}{z_j}$ for which $\{z_j\}_{j\in \mathbb Z_0}$ is the solution of the system of equations (\ref{e11}).

\textbf{Solution.} First we find the sum of the series $\sum_{j\in\mathbb Z_0}{\lambda_j}$.
 $$\Lambda=\sum_{j\in \mathbb Z_0}{\lambda_j}=\sum_{j=1}^\infty{\frac{2.4^j}{j!}\cdot e^{-2.4}}+\sum_{j=1}^{\infty}\frac{8^j}{j!}\cdot e^{-8}=\Pi(2.4)+\Pi(8)=0.9003+0.9997=1.9.$$

For $k=2$ and $\Lambda=1.9$ from the equality $\Lambda=A(1+A)^2$ we find $A_0\approx0.676223$. We define terms of the series from the equality $z_i=\frac{\lambda_i}{(1+A_0)^2}$:
 $$z_1=\frac{2.4}{e^{2.4}}\cdot\frac{1}{(1+A_0)^2}\approx0.07748914,\ z_2\approx0.0929869, \ \ldots, \ z_n=\frac{2.4^n}{e^{2.4}n!}\cdot\frac{1}{(1+A_0)^2}, \ \ldots,$$
$$z_{-1}=\frac{8}{e^{8}}\cdot\frac{1}{(1+A_0)^2}\approx0.0009551476,\ z_{-2}\approx0.003820590, \ \ldots, \ z_{-n}=\frac{8^n}{e^{8}n!}\cdot\frac{1}{(1+A_0)^2}, \ldots.$$
Hence the desired series $\sum_{j\in \mathbb Z_0}{z_j}$ has the following form:
$$\sum_{j\in \mathbb Z_0}{z_j}=\frac{1}{(1+A_0)^2}\cdot\sum_{n=-1}^{-\infty}\frac{8^{-n}}{e^{8}\mid n\mid!}+\frac{1}{(1+A_0)^2}\cdot\sum_{n=1}^\infty\frac{2.4^n}{e^{2.4}n!}=A_0.$$

\textbf{Example 3.} Let $k=2$. For a sequence of parameters given by a geometric distribution:
 $$\lambda_j=\left\{%
\begin{array}{ll}
\alpha(1-\alpha)^j, \ \  \alpha \in (0;1), \ \ j\in \mathbb N, \\ [3mm]
\beta(1-\beta)^{-j}, \ \ \beta\in (0;1),\ \ \alpha+\beta=0.875, \ \ j\in \mathbb {Z}_0\setminus \mathbb N.
\end{array}%
\right.$$
find a series $\sum_{j\in \mathbb Z_0}{z_j}$ for which $\{z_j\}_{j\in \mathbb Z_0}$ is the solution of the system of equations (\ref{e11}).

\textbf{Solution.} First we find the sum of the series $\sum_{j\in\mathbb Z_0}{\lambda_j}$.
$$\Lambda=\sum_{j\in \mathbb Z_0}{\lambda_j}=\sum_{j=1}^\infty{\alpha(1-\alpha)^j}+\sum_{j=1}^{\infty}{\beta(1-\beta)^{j}}=
1-\alpha+1-\beta=2-(\alpha+\beta)=1.125.$$

For $k=2$ and $\Lambda=1.125$ from the equality $\Lambda=A(1+A)^2$ we find $A_0=0.5$. From the equality $z_i=\frac{\lambda_i}{(1+A_0)^2}$ we define terms of the series:
$$z_1=\frac{4\alpha(1-\alpha)}{9}, \  z_2=\frac{4\alpha(1-\alpha)^2}{9},\ \ldots, \ z_n=\frac{4\alpha(1-\alpha)^n}{9}, \ \ldots,$$
$$z_{-1}=\frac{4\beta(1-\beta)}{9}, \ z_{-2}=\frac{4\beta(1-\beta)^2}{9}, \ \ldots, \ z_{-n}=\frac{4\beta(1-\beta)^n}{9}, \ \ldots.$$
So the desired series $\sum_{j\in \mathbb Z_0}{z_j}$ has the following form:
$$\sum_{j\in \mathbb Z_0}{z_j}=\sum_{n=-1}^{-\infty}\frac{4\beta(1-\beta)^{-n}}{9}+\sum_{n=1}^\infty\frac{4\alpha(1-\alpha)^n}{9}=
\frac12.$$

\section{Periodic Gibbs measures}

It is known that there exists one-to-one correspondence between the set $V$ of vertices of a Cayley tree
of order $k\geq1$ and the group $G_k$ that is the free product of $k+1$ cyclic groups of second order with the
corresponding generators $a_1,a_2, \ldots, a_{k+1}$. Therefore, the set $V$ can be identified with the set $G_k$.

Let $G_k/\widehat{G}_k=\{H_1,...,H_r\}$ be the quotient group, where
$\widehat{G}_k$ is a normal subgroup of index $r\geq 1.$

\begin{defn}  The set of vectors $z=\{z_x,x\in G_k\}$
is said to be $\widehat{G}_k$- periodic if  $z_{yx}=z_x$  for all
$\forall x\in G_k, y\in\widehat{G}_k.$\

$G_k$-periodic sets are said to be translation-invariant.
\end{defn}
\begin{defn}  A measure $\mu$ is said to be
$\widehat{G}_k$-periodic if it corresponds to the
$\widehat{G}_k$-periodic set of vectors $z$.
\end{defn} 
Let $G^{(2)}_k$ be the subgroup of $G_k$ consisting the words of even length. 

Consider the $G^{(2)}_k$-periodic Gibbs measures that correspond to set of vectors $z=\{z_x\in \mathbb R_{+}^{\infty}: \, x\in G_k\}$ of the form
$$z_x=\left\{%
\begin{array}{ll}
    z, \ \ \ $ if $ |x| \, \mbox{is even} $,$ \\
    \widetilde{z}, \ \ \ $ if $ |x| \, \mbox{is odd} $.$ \\
\end{array}%
\right. $$
Here $z=(\ldots,\ z_{-1},z_{0}, \ z_{1},\ldots),$ $\widetilde{z}=(\ldots,\ \widetilde{z}_{-1},\widetilde{z}_{0}, \ \widetilde{z}_{1},\ldots).$

Then due to (\ref{e10}) for $G^{(2)}_k$-periodic Gibbs measures we have:
\begin{equation}\label{e14}\left\{%
\begin{array}{ll}
z_i=\lambda_i\cdot \left({1\over 1+\sum_{j\in \mathbb Z_0} \widetilde{z}_j}\right)^k, \ \ i\in \mathbb Z_0, \\
\widetilde{z}_i=\lambda_i\cdot \left({1\over 1+\sum_{j\in \mathbb Z_0} z_j}\right)^k, \ \ i\in \mathbb Z_0.
\end{array}%
\right.
\end{equation}

\begin{rk} We note that the solution $(z,z)$ (i.e., $z=\widetilde{z}$) in (\ref{e14}) corresponds to the unique translation-invariant Gibbs measure. Therefore, we are interested in solutions of the form  $(z,\widetilde{z})$,  $z\neq\widetilde{z}$.
\end{rk}
\begin{lemma}\label{L2} Let $k\geq2$. If there is a positive solution $(z_j, \widetilde{z}_j), \ j\in \mathbb Z_0$ of the system of equations (\ref{e14}) for some sequence of parameters $\{\lambda_j\}_{j\in \mathbb Z_0}$ then the series $\sum_{j\in \mathbb Z_0} z_j$, $\sum_{j\in \mathbb Z_0} \widetilde{z}_j$ and $\sum_{j\in \mathbb Z_0}\lambda_j$  converge.
\end{lemma}
\begin{proof} Let $k\geq2$  and there is a positive solution $(z_j, \widetilde{z}_j), \ j\in \mathbb Z_0$ of the system of equations (\ref{e14}). Let us first show that the series $\sum_{j\in \mathbb Z_0} z_j$ and $\sum_{j\in \mathbb Z_0} \widetilde{z}_j$ converge. Assume the opposite, let one of the series, for example, the series $\sum_{j\in \mathbb Z_0} z_j$ diverge, i.e., $\sum_{j\in \mathbb Z_0} z_j=+\infty$. Then from the second equation of (\ref{e14}) we obtain $\widetilde{z}_j=0$, i.e., the system of equations (\ref{e14}) has no positive solutions. In the case when both series $\sum_{j\in \mathbb Z_0} z_j$, $\sum_{j\in \mathbb Z_0} \widetilde{z}_j$ diverge, it is obvious that the system of equations (\ref{e14}) has no solutions. Hence if there is a solution (\ref{e14}) then the series $\sum_{j\in \mathbb Z_0} z_j$ and $\sum_{j\in \mathbb Z_0} \widetilde{z}_j$ converge.

Now let is prove that the series $\sum_{j\in \mathbb Z_0}\lambda_j$ converge. For this we assume  $\sum_{j\in \mathbb Z_0} z_j=A$ and $\sum_{j\in \mathbb Z_0} \widetilde{z}_j=B$. Then from (\ref{e14}) we obtain $\sum_{i\in \mathbb Z_0}\lambda_i=A(1+B)^k$.
Lemma is proved.
\end{proof}
\begin{rk} It follows from Lemma \ref{L2} that if one of the series $\sum_{j\in \mathbb Z_0} z_j$, $\sum_{j\in \mathbb Z_0} \widetilde{z}_j$ and $\sum_{j\in \mathbb Z_0}\lambda_j$ diverges, then the system of equations (\ref{e14}) does not have a positive solution.
\end{rk}
Let $\sum_{j\in \mathbb Z_0} z_j=A$ and $\sum_{j\in \mathbb Z_0} \widetilde{z}_j=B$. Then from (\ref{e14}) we have
$$z_i=\frac{\lambda_i}{(1+B)^k}, \ \ \ \widetilde{z}_i=\frac{\lambda_i}{(1+A)^k}.$$
Consequently,
$$\sum_{j\in \mathbb Z_0}z_j=\frac{\sum_{j\in \mathbb Z_0}\lambda_j}{(1+B)^k}, \ \ \sum_{j\in \mathbb Z_0}\widetilde{z}_j=\frac{\sum_{j\in \mathbb Z_0}\lambda_j}{(1+A)^k}$$
or
\begin{equation}\label{e15}
A(1+B)^k=\sum_{j\in \mathbb Z_0}\lambda_j, \ \ \ \ B(1+A)^k=\sum_{j\in \mathbb Z_0}\lambda_j.
\end{equation}

\begin{pro} Let $k\geq2$. If the series $\sum_{j\in \mathbb Z_0} z_j$ converge and its sum is $\sum_{j\in \mathbb Z_0} z_j=A\neq\frac{1}{k-1}$. Then there exists a number $B$ which is the sum of a unique descending series  $\sum_{j\in \mathbb Z_0} \widetilde{z}_j$ $(z\neq \widetilde{z})$ under the condition $A(1+B)^k=B(1+A)^k$ and there is a unique sequence of parameters $\{\lambda_j\}_{j\in \mathbb Z_0}$ such that  $(z,\widetilde{z})$ and $(\widetilde{z},z)$ are solutions of (\ref{e14}), where $z=\{z_j\}_{j\in \mathbb Z_0},$ $\widetilde{z}=\{\widetilde{z}_j\}_{j\in \mathbb Z_0}$.
\end{pro}
\begin{proof} Let $\sum_{j\in \mathbb Z_0} z_j=A$. First we prove the existence of the number $B$ under the conditions $A(1+B)^k=B(1+A)^k$, $z\neq \widetilde{z}$, ($A\neq B$) and $A\neq\frac{1}{k-1}$.

We note that for $A=B$ by (\ref{e14}) it follows that $z_j= \widetilde{z}_j$ for any $j\in \mathbb Z_0$, i.e., the solution of this form corresponds to the translation-invariant Gibbs measure. Therefore, we will consider the case $A\neq B$.

In the equality
\begin{equation}\label{e00}
A(1+B)^k-B(1+A)^k=0
\end{equation}
we introduce the notation $B=x, \ A=y$ and (for a fixed $y$)    consider the following function:
$$f(x)=y(1+x)^k-x(1+y)^k.$$
It is clear that $f(y)=0$, i.e., $x=y$ is the root of the equation $f(x)=0$.

We consider the case $x\neq y$. We rewrite the function $f(x)$ as follows:
$$f(x)=y\Big(x^k+kx^{k-1}+\frac{k(k-1)}{2}x^{k-2}+...+\frac{k(k-1)}{2}x^{2}\Big)-
\Big(\Big(1+y\Big)^k-ky\Big)x+y.$$

From the Bernoulli inequality we have $(1+y)^k>ky, \ (y>0).$ Then on the RHS of the last equality the signs change twice, i.e., by the Descartes rule of signs, the equation $f(x)=0$ has two positive root or this equation has no positive solutions at all. But we have the solution $x=y$. So, the equation $f(x)=0$ has one more positive solution different from $x=y$.

Next we check the multiplicity of the root $x=y$. For this we use the theorem on zeros of the holomorphic function, i.e., $x=y$ is not a multiple root of the equation $f(x)=0$, if $f(y)=0$ and $f'(y)\neq 0$.

We have $f(y)=0$ and $f'(x)=ky(1+x)^{k-1}-(1+y)^k$. Therefore, from
$$ f'(y)=(1+y)^{k-1}\Big(\Big(k-1\Big)y-1\Big)=0$$
we get 
$$y=\frac{1}{k-1}.$$
Hence $x=y$ is not a multiple root, if $y\neq \frac{1}{k-1}$.

Let $\widetilde{B}$ be a solution of (\ref{e00}) different from $A$. Then the existence and uniqueness of $\lambda_i$ follows from the equality $\lambda_i=z_i(1+\widetilde{B})^k, \ i\in \mathbb Z_0$.

There are infinitely many sequences $\{\widetilde{z}_j\}_{j\in \mathbb Z_0}$ for which $\sum_{j\in \mathbb Z_0} \widetilde{z}_j=\widetilde{B}$.
Among them the sequence $\{\widetilde{z}_j\}_{j\in \mathbb Z_0}$ for which there exists $\{\lambda_j\}_{j\in \mathbb Z_0}$ such that $\{\widetilde{z}_j\}_{j\in \mathbb Z_0}$ satisfies (\ref{e14}) is unique. This follows from the equality $\widetilde{z}_i=\frac{\lambda_i}{(1+A)^k}$.

From the above it follows that there is a sequence of parameters $\{\lambda_j\}_{j\in \mathbb Z_0}$ for which $(z,\widetilde{z})$ is the solution of  (\ref{e14}), where $z=\{z_j\}_{j\in \mathbb Z_0}$ and $\widetilde{z}=\{\widetilde{z}_j\}_{j\in \mathbb Z_0}$. We note that due to the symmetry  $(\widetilde{z}, z)$ is also a solution of (\ref{e14}). The proof is complete.\
\end{proof}
Now let us determine under what conditions on the parameters $\{\lambda_j\}_{j\in \mathbb Z_0}$ there are solutions  $(z,\widetilde{z})$ and $(\widetilde{z}, z)$. We introduce the following definition.

\begin{defn}\cite{Kel} A twice continuously differentiable function $f:[0, \infty) \mapsto [0, \infty)$  is
said to be $S$-shaped if it has the following properties:
\begin{itemize}
\item[(1)] It is increasing on $[0, \infty)$ with $f(0)>0$ and $\sup_x{f (x)}<\infty$;

\item[(2)] There exists $\widetilde{x}\in (0,\infty)$ such that the derivative $f'$  is monotone increasing in the interval $(0, \widetilde{x})$ and monotone decreasing in the interval $(\widetilde{x},\infty)$; in other words, $\widetilde{x}$ satisfies $f''(\widetilde{x})=0$ and is the unique inflection point of $f(x)$.
\end{itemize}
\end{defn}
It is known \cite{Kel} that any $S$-shaped function has at most three fixed points in the interval $[0, \infty)$.

\begin{lemma}\cite{K}\label{LK} Let $f:[0,1]\rightarrow [0,1]$
be a continuous function with a fixed point $\xi \in (0,1)$. We
assume that $f$ is differentiable at $\xi$ and $f^{'}(\xi)<-1.$
Then there exist points $x_0$ and $x_1$, $ 0\leq x_0<\xi<x_1
\leq1,$ such that $f(x_0)=x_1$ and $f(x_1)=x_0.$
\end{lemma}

We rewrite the system of equations (\ref{e15}):
\begin{equation}\label{e16}\left\{%
\begin{array}{ll}
A=f(B); \\
B=f(A).
\end{array}%
\right.
\end{equation}
Here $f(x)=\frac{\Lambda}{(1+x)^k}.$ The system of equations (\ref{e16}) is well studied in \cite{Kel} (sec. 2.2, p. 904) and \cite{Mar} (sec. 5.2, p. 153). It is shown that the function $h(x) =f(f(x))$ is an $S$-shaped function. The following properties of the function $h(x)$ are obvious:
\begin{itemize}
\item  $h(x)$ is an $S$-shaped function with $h(0)=\frac{\Lambda}{(1+\Lambda)^k}$ and $\sup_x{h(x)=\Lambda}.$

\item $f(x)$  has a unique fixed point,  $x_0$,  which is also a fixed point of  $h(x)$.

\item There exists $\Lambda_{\rm cr}>0$ such that if $\Lambda\leq \Lambda_{\rm cr}$ then $h'(x)\leq1$ for any $x\geq0$ and $x_0$ is the only fixed point for $h(x)$.

\item If $\Lambda>\Lambda_{\rm cr},$ then $h(x)$ has three fixed points $x_1<x_0<x_2$, where $f(x_1)=x_2$ and $f(x_2)=x_1$. Moreover $h'(x_0)>1,$ $h'(x)<1$ for $x\in [0;x_1]\cup[x_2;\infty)$ and the three fixed points converge to $x_0( \Lambda_{\rm cr})$ as $\Lambda\rightarrow\Lambda_{\rm cr}.$
\end{itemize}
Since $h'(x_0)=(f'(x_0))^2$, it is easy to see that $x_0$ is the unique fixed point of the function $h(x)$ if and only if $f'(x_0)\geq -1$.

Let $A_0$ be the unique solution of the equation $A=\frac{\Lambda}{(1+A)^k}$. We calculate the derivative $f'(A_0)$:
$$f'(A)=-\frac{k\Lambda}{(1+A)^{k+1}}, \ \ \ f'(A_0)=-\frac{kA_0}{1+A_0}.$$
We solve the inequality $f'(A_0)\geq-1$, and its solution has the form $A_0\leq\frac{1}{k-1}$.
Then from $A_0=\frac{\Lambda}{(1+A_0)^k}$ we obtain that the equation $h(x)=x$ has only one fixed point for
$$\Lambda\leq\frac{k^k}{(k-1)^{k+1}}=\Lambda_{\rm cr}(k).$$
From the inequality $f'(A_0)<-1$ we can get $\Lambda>\Lambda_{\rm cr}$. Then by Lemma \ref{LK}, there are at least three fixed points. On the other side under this condition by the property of an $S$-shaped function we have at most three fixed points for the function $h(x)$. Hence there are exactly three fixed points of the equation $h(x)=x$ for $\Lambda>\Lambda_{\rm cr}$.

So the system of equations (\ref{e16}) has a unique solution of the form $(A; A)$, i.e., $A=B$ for $\Lambda\leq\Lambda_{\rm cr}$, and for $\Lambda>\Lambda_{\rm cr}$ it has two positive solutions $\Big(A_0; B_0\Big)$ and $\Big(B_0; A_0\Big)$ besides the solution $(A, A)$.

\begin{rk} Solutions of the system of equations (\ref{e14}) corresponding to $G^{(2)}_k$-periodic Gibbs measures are sometimes called two-periodic.
Similarly to work \cite{HKR} it is easy to obtain that two-periodic solutions $(z,\widetilde{z})$ and $(\widetilde{z}, z)$ $\Big(z=\{z_j\}_{j\in \mathbb Z_0}, \ \widetilde{z}=\{\widetilde{z}_j\}_{j\in \mathbb Z_0}\Big)$ of the system of equations (\ref{e14}) are normalizable if the series $\sum_{i\in \mathbb Z}z_i^{\frac{k+1}{k}}$ and $\sum_{i\in \mathbb Z}\widetilde{z}_i^{\frac{k+1}{k}}$ converge. These series by virtue of Lemma \ref{L2} (a necessary condition for the existence of a solution), converge. Then in the case of two-periodic solutions there exist Gibbs measures $\mu_1$ and $\mu_2$ corresponding to the solutions $(z,\widetilde{z})$ and $(\widetilde{z}, z)$, respectively.
\end{rk}

\textbf{Markov chain corresponding to the periodic Gibbs measure.}

Similarly to the case of the translation-invariant Gibbs measure, we check the existence of a stationary distribution of the Markov chain corresponding to the measures $\mu_1$ and $\mu_2$. Using the method from \cite{RKhM} ($\mathbb P=\mathbb P_{\mu_1}\cdot \mathbb P_{\mu_2}$) we construct the probability transition matrix $\mathbb P$ corresponding to the measure $\mu_i$ $(i=1,2)$:
$$\mathbb P=\begin{pmatrix}
 & \vdots & \vdots & \vdots & \vdots & \vdots & \\
 \cdots & 0 & 0 & 1  &  0 &  0 & \cdots \\  \cdots & 0 & 0 & 1  &  0 &  0 & \cdots \\
 \cdots & \frac{\lambda_{-2} z_{-2}}{1+\sum_{l \in \mathbb Z_0}{\lambda_lz_l}} & \frac{\lambda_{-1} z_{-1}}{1+\sum_{l \in \mathbb Z_0}{\lambda_lz_l}} &  \frac{1}{1+\sum_{l \in \mathbb Z_0}{\lambda_lz_l}}  &  \frac{\lambda_1 z_1}{1+\sum_{l \in \mathbb Z_0}{\lambda_lz_l}}&  \frac{\lambda_2 z_2}{1+\sum_{l \in \mathbb Z_0}{\lambda_lz_l}} & \cdots \\
 \cdots & 0 & 0 & 1  &  0 &  0 & \cdots \\  \cdots & 0 & 0 & 1  &  0 &  0 & \cdots \\
  & \vdots & \vdots & \vdots & \vdots & \vdots &  \end{pmatrix}\times$$
 $$\times  \begin{pmatrix}
  & \vdots & \vdots & \vdots & \vdots & \vdots &  \\
 \cdots & 0 & 0 & 1  &  0 &  0 & \cdots \\  \cdots & 0 & 0 & 1  &  0 &  0 & \cdots \\
 \cdots & \frac{\lambda_{-2} \widetilde{z}_{-2}}{1+\sum_{l \in \mathbb Z_0}{\lambda_l\widetilde{z}_l}} & \frac{\lambda_{-1} \widetilde{z}_{-1}}{1+\sum_{l \in \mathbb Z_0}{\lambda_l\widetilde{z}_l}} &  \frac{1}{1+\sum_{l \in \mathbb Z_0}{\lambda_l\widetilde{z}_l}}  &  \frac{\lambda_1 \widetilde{z}_1}{1+\sum_{l \in \mathbb Z_0}{\lambda_l\widetilde{z}_l}}&  \frac{\lambda_2 \widetilde{z}_2}{1+\sum_{l \in \mathbb Z_0}{\lambda_l\widetilde{z}_l}} & \cdots \\
 \cdots & 0 & 0 & 1  &  0 &  0 & \cdots \\  \cdots & 0 & 0 & 1  &  0 &  0 & \cdots \\
 & \vdots & \vdots & \vdots & \vdots & \vdots & 
 \end{pmatrix}.$$
 $$\mathbb P=\begin{pmatrix}
  & \vdots & \vdots & \vdots & \vdots & \vdots & \\
 \cdots & \frac{\lambda_{-2} \widetilde{z}_{-2}}{1+\sum_{l \in \mathbb Z_0}{\lambda_l\widetilde{z}_l}} & \frac{\lambda_{-1} \widetilde{z}_{-1}}{1+\sum_{l \in \mathbb Z_0}{\lambda_l\widetilde{z}_l}} &  \frac{1}{1+\sum_{l \in \mathbb Z_0}{\lambda_l\widetilde{z}_l}}  &  \frac{\lambda_1 \widetilde{z}_1}{1+\sum_{l \in \mathbb Z_0}{\lambda_l\widetilde{z}_l}}&  \frac{\lambda_2 \widetilde{z}_2}{1+\sum_{l \in \mathbb Z_0}{\lambda_l\widetilde{z}_l}} & \cdots \\
 \cdots & \frac{\lambda_{-2} \widetilde{z}_{-2}}{1+\sum_{l \in \mathbb Z_0}{\lambda_l\widetilde{z}_l}} & \frac{\lambda_{-1} \widetilde{z}_{-1}}{1+\sum_{l \in \mathbb Z_0}{\lambda_l\widetilde{z}_l}} &  \frac{1}{1+\sum_{l \in \mathbb Z_0}{\lambda_l\widetilde{z}_l}}  &  \frac{\lambda_1 \widetilde{z}_1}{1+\sum_{l \in \mathbb Z_0}{\lambda_l\widetilde{z}_l}} &  \frac{\lambda_2 \widetilde{z}_2}{1+\sum_{l \in \mathbb Z_0}{\lambda_l\widetilde{z}_l}} & \cdots \\
 \cdots & c_{0-2} & c_{0-1} &  c_{00}  &  c_{01}&  c_{02} & \cdots \\
 \cdots & \frac{\lambda_{-2} \widetilde{z}_{-2}}{1+\sum_{l \in \mathbb Z_0}{\lambda_l\widetilde{z}_l}} & \frac{\lambda_{-1} \widetilde{z}_{-1}}{1+\sum_{l \in \mathbb Z_0}{\lambda_l\widetilde{z}_l}} &  \frac{1}{1+\sum_{l \in \mathbb Z_0}{\lambda_l\widetilde{z}_l}}  &  \frac{\lambda_1 \widetilde{z}_1}{1+\sum_{l \in \mathbb Z_0}{\lambda_l\widetilde{z}_l}}&  \frac{\lambda_2 \widetilde{z}_2}{1+\sum_{l \in \mathbb Z_0}{\lambda_l\widetilde{z}_l}} & \cdots \\
 \cdots & \frac{\lambda_{-2} \widetilde{z}_{-2}}{1+\sum_{l \in \mathbb Z_0}{\lambda_l\widetilde{z}_l}} & \frac{\lambda_{-1} \widetilde{z}_{-1}}{1+\sum_{l \in \mathbb Z_0}{\lambda_l\widetilde{z}_l}} &  \frac{1}{1+\sum_{l \in \mathbb Z_0}{\lambda_l\widetilde{z}_l}}  &  \frac{\lambda_1 \widetilde{z}_1}{1+\sum_{l \in \mathbb Z_0}{\lambda_l\widetilde{z}_l}}&  \frac{\lambda_2 \widetilde{z}_2}{1+\sum_{l \in \mathbb Z_0}{\lambda_l\widetilde{z}_l}} & \cdots \\
  & \vdots & \vdots & \vdots & \vdots & \vdots & 
 \end{pmatrix}.$$
Here  we have
$$c_{0i}=\frac{1}{1+\sum_{l \in \mathbb Z_0}{\lambda_lz_l}}\cdot\frac{\lambda_{i} \widetilde{z}_{i}}{1+\sum_{l \in \mathbb Z_0}{\lambda_l\widetilde{z}_l}}, \ \ c_{00}=1-\frac{1}{1+\sum_{l \in \mathbb Z_0}{\lambda_lz_l}}\cdot\frac{\sum_{l \in \mathbb Z_0}{\lambda_l\widetilde{z}_l}}{1+\sum_{l \in \mathbb Z_0}{\lambda_l\widetilde{z}_l}}$$
for $i\in \mathbb Z_0$. Consider the vector $X=(\cdots,x_{-2},x_{-1},x_0,x_1,x_2,\cdots), \ \sum_{j\in \mathbb Z}{x_j}=1$ and we solve the system of equations $X\cdot\mathbb P=X$. Then
\begin{equation}\label{e18}
x_0=\frac{1-x_0}{1+\sum_{l \in \mathbb Z_0}{\lambda_l\widetilde{z}_l}}+x_0\cdot c_{00}, \ \ x_j=\frac{(1-x_0)\lambda_{j}\widetilde{z}_{j}}{1+\sum_{l \in \mathbb Z_0}{\lambda_l\widetilde{z}_l}}+x_0\cdot c_{0j}, \  \ j\in \mathbb Z_0.
\end{equation}
From the first equality of (\ref{e18}) we find $x_0$:
$$x_0=\frac{1+\sum_{l \in \mathbb Z_0}{\lambda_lz_l}}{1+\sum_{l \in \mathbb Z_0}{\lambda_l\widetilde{z}_l}+\sum_{l \in \mathbb Z_0}{\lambda_lz_l}}.$$
Using the expression for $x_0$ from the second equality (\ref{e18}) we find $x_j, \ j\in\mathbb Z_0$:
$$x_j=\frac{\lambda_{j}\widetilde{z}_{j}}{{1+\sum_{l \in \mathbb Z_0}{\lambda_l\widetilde{z}_l}+\sum_{l \in \mathbb Z_0}{\lambda_lz_l}}}.$$
It is easy to see that the resulting vector is stochastic:
$$\sum_{j\in \mathbb Z}{x_j}=\frac{\sum_{l \in \mathbb Z_0}\lambda_{j}\widetilde{z}_{j}}{1+\sum_{l \in \mathbb Z_0}\lambda_{j}\widetilde{z}_{j}+\sum_{l \in \mathbb Z_0}\lambda_{j}z_{j}}+\frac{1+\sum_{l \in \mathbb Z_0}\lambda_{j}z_{j}}{1+\sum_{l \in \mathbb Z_0}\lambda_{j}\widetilde{z}_{j}+\sum_{l \in \mathbb Z_0}\lambda_{j}z_{j}}=1.$$
Hence there exists a stationary distribution of the Markov chain corresponding to the measure $\mu_i \ (i=1,2)$.

Thus, the following theorem is true.

\begin{thm} Let $k\geq2$ and $\Lambda_{\rm cr}(k)=\frac{k^k}{(k-1)^{k+1}}$. Then for the HC model with a countable set of states  (corresponding to the graph from Fig.\ref{fi}) the following statements are true:
\begin{itemize}
\item[1.] If the series $\sum_{j\in \mathbb Z_0}\lambda_j$ obtained from a sequence of parameters $\{\lambda_j\}_{j\in \mathbb Z_0}$ converges and its sum is $\sum_{j\in \mathbb Z_0}\lambda_j=\Lambda$, then for $0<\Lambda\leq\Lambda_{\rm cr}$ there exists unique $G_k^{(2)}$-periodic Gibbs measure $\mu_0$ that is translation-invariant, and for $\Lambda>\Lambda_{\rm cr}$ there are exactly three $G_k^{(2)}$-periodic Gibbs measures $\mu_0, \mu_1, \mu_2$, where measures $\mu_1$ and $\mu_2$ are $G_k^{(2)}$-periodic(non translation-invariant) Gibbs measures.
\item[2.] If the series $\sum_{j\in \mathbb Z_0}\lambda_j$  diverges then there is no $G_k^{(2)}$-periodic Gibbs measure.
\end{itemize}
\end{thm}

By virtue of (\ref{e14}) from equalities $z_i=\frac{\lambda_i}{(1+B_0)^k}$ and $\widetilde{z}_i=\frac{\lambda_i}{(1+A_0)^k}$ we can get
$$\sum_{j\in \mathbb Z_0} \frac{\lambda_j}{(1+B_0)^k}=A_0,\ \ \ \sum_{j\in \mathbb Z_0} \frac{\lambda_i}{(1+A_0)^k}=B_0.$$

\begin{rk} We have $\Lambda_{cr}(2)=4$ for $k=2$ and the sums of the series $\sum_{j\in \mathbb Z_0}z_j$ and $\sum_{j\in \mathbb Z_0}\widetilde{z}_j$ obtained from the solution $(z,\widetilde{z})$ have the following sum
\begin{equation}\label{e19}
A_0(\Lambda)=\frac{\Lambda-2\pm\sqrt{\Lambda(\Lambda-4)}}{2}, \ B_0(\Lambda)=\frac{\Lambda-2\mp\sqrt{\Lambda(\Lambda-4)}}{2}.
\end{equation}
\end{rk}
\textbf{Example 4.} Let $k=2$. Find solutions of the system of equations (\ref{e14}) corresponding to $G_{k}^{(2)}$-periodic (not translation-invariant) Gibbs measures for the sequence
$$\lambda_j=\left\{%
\begin{array}{ll}
\frac{9}{(4j-3)(4j-1)}, \ \ j\in \mathbb N, \\
\frac{9}{(4j-1)(4j+1)}, \ \ -j\in \mathbb N.
\end{array}%
\right.$$

\textbf{Solution.} First we find the sum of the series $\sum_{j\in Z_0}{\lambda_j}$:
$$\Lambda=\sum_{j\in Z_0}{\lambda_j}=9\cdot\sum_{n=1}^\infty{\frac{1}{(2n-1)(2n+1)}}=9\Big(\frac{1}{1\cdot3}+\frac{1}{3\cdot5}+\frac{1}{5\cdot7}+\ldots\Big)=$$
$$=\frac92\cdot\Big(\frac11-\frac13+\frac13-\frac15+\frac15-\frac17+\frac17+\ldots\Big)=\frac92.$$

There exist solutions (with condition $A\neq B$) of the system of equations (\ref{e14}) corresponding to $G_{k}^{(2)}$-periodic (not translation-invariant) Gibbs measures for $k=2$ and $\Lambda>\Lambda_{\rm cr}(2)=4$.
Using the formulas (\ref{e19}) we solve the following system of equations
$$\left\{%
\begin{array}{ll}
A(1+B)^2=\frac92; \\ [2mm]
B(1+A)^2=\frac92.
\end{array}%
\right.$$
Solutions have the following forms: ($\frac12,2$) and ($2,\frac12$).

Let $A_0=\frac12$ and $B_0=2$. In this case, we determine the terms of the series from the equality $z_i=\frac{\lambda_i}{(1+B_0)^2}$:
$$z_x=\Big(...,\ \frac{1}{(4n-1)(4n+1)},\ \ldots, \ \frac{1}{63}, \ \frac{1}{15}, \ 1, \ \frac{1}{3}, \ \frac{1}{35}, \ \ldots, \ \frac{1}{(4n-3)(4n-1)}, \ \ldots \Big)$$
Hence, the desired series $\sum_{j\in \mathbb Z_0}{z_j}$ has the following form
$$\sum_{j\in \mathbb Z_0}{z_j}=\sum_{n=-1}^{-\infty}\frac{1}{(4n+1)(4n-1)}+\sum_{n=1}^\infty\frac{1}{(4n-3)(4n-1)}=\frac12.$$
We determine the values of $\widetilde{z}_j$ from the equality $\widetilde{z}_j=\frac{\lambda_j}{(1+A_0)^2}$:
$$\widetilde{z}_x=\Big(\ldots, \ \frac{4}{(4n-1)(4n+1)}, \ \ldots, \ \frac{4}{63}, \ \frac{4}{15}, \ 1, \ \frac{4}{3}, \ \frac{4}{35}, \ \ldots, \ \frac{4}{(4n-3)(4n-1)}, \ \ldots\Big)$$
From here the series $\sum_{j\in \mathbb Z_0}{\widetilde{z}_j}$ has the form
$$\sum_{j\in \mathbb Z_0}{\widetilde{z}_j}=\sum_{n=-1}^{-\infty}\frac{4}{(4n+1)(4n-1)}+\sum_{n=1}^\infty\frac{4}{(4n-3)(4n-1)}=2.$$

In the case $A_0=2$, $B_0=\frac12$ we can similarly obtain as solutions
$$z_x=\Big(\ldots, \ \frac{4}{(4n-1)(4n+1)}, \ \ldots, \ \frac{4}{63}, \ \frac{4}{15}, \ 1, \ \frac{4}{3}, \ \frac{4}{35}, \ \ldots, \ \frac{4}{(4n-3)(4n-1)}, \ \ldots\Big),$$
$$\widetilde{z}_x=\Big(...,\ \frac{1}{(4n-1)(4n+1)},\ \ldots, \ \frac{1}{63}, \ \frac{1}{15}, \ 1, \ \frac{1}{3}, \ \frac{1}{35}, \ \ldots, \ \frac{1}{(4n-3)(4n-1)}, \ \ldots \Big)$$
and
$$\sum_{j\in \mathbb Z_0}{z_j}=2,\ \ \sum_{j\in \mathbb Z_0}{\widetilde{z}_j}=1/2.$$\

\textbf{Example 5.} Let $k=2$. For a sequence of parameters given by a geometric distribution
$$ \lambda_j=\left\{%
\begin{array}{ll}
4\alpha(1-\alpha)^j, \ \  \alpha \in (0;1), \ \ j\in \mathbb N, \\ [3mm]
4\beta(1-\beta)^{-j}, \ \ \beta\in (0;1),\ \ \alpha+\beta=\frac23, \ \ j\in \mathbb {Z}_0\setminus \mathbb N.
\end{array}
\right.$$
find solutions of the system of equations (\ref{e14}) corresponding to $G_{k}^{(2)}$-periodic (not translation-invariant) Gibbs measures.

\textbf{Solution.} First we find the sum of the series $\sum_{j\in Z_0}{\lambda_j}$:
$$\Lambda=\sum_{j\in \mathbb Z_0}{\lambda_j}=4\sum_{j=1}^\infty{\alpha(1-\alpha)^j}+4\sum_{j=1}^{\infty}{\beta(1-\beta)^{j}}=4\Big(1-\alpha+1-\beta\Big)=\frac{16}{3}.$$

There exist solutions (with condition $A\neq B$) of the system of equations (\ref{e14}) corresponding to $G_{k}^{(2)}$-periodic (not translation-invariant) Gibbs measures for $k=2$ and $\Lambda>\Lambda_{\rm cr}(2)=4$.
Using the formulas (\ref{e19}) we solve the following system of equations

$$\left\{%
\begin{array}{ll}
A(1+B)^2=\frac{16}{3}; \\ [2mm]
B(1+A)^2=\frac{16}{3}.
\end{array}%
\right.$$
Solutions have the following forms: ($\frac13,3$) and ($3,\frac13$).

Let $A_0=\frac13$ and $B_0=3$. In this case, we determine the terms of the series from the equality $z_i=\frac{\lambda_i}{(1+B_0)^2}$.
Hence, the desired series $\sum_{j\in \mathbb Z_0}{z_j}$ has the following form:
$$\sum_{j\in \mathbb Z_0}{z_j}=\sum_{j=1}^\infty\frac{\alpha(1-\alpha)^j}{4}+\sum_{j=-1}^{-\infty}\frac{\beta(1-\beta)^{-j}}{4}
=\frac13.$$
We determine the values of $\widetilde{z}_j$ from the equality  $\widetilde{z}_j=\frac{\lambda_j}{(1+A_0)^2}$.
From here the series $\sum_{j\in \mathbb Z_0}{\widetilde{z}_j}$ has the form:
$$\sum_{j\in \mathbb Z_0}{\widetilde{z}_j}=
\sum_{j=1}^\infty\frac{9\alpha(1-\alpha)^j}{4}+\sum_{j=-1}^{-\infty}\frac{9\beta(1-\beta)^{-j}}{4}=3.$$

In the case $A_0=3$, $B_0=\frac13$ we can similarly obtain as solutions
$$\sum_{j\in \mathbb Z_0}{z_j}=
\sum_{j=1}^\infty\frac{9\alpha(1-\alpha)^j}{4}+\sum_{j=-1}^{-\infty}\frac{9\beta(1-\beta)^{-j}}{4}=3,$$
$$\sum_{j\in \mathbb Z_0}{\widetilde{z}_j}=\sum_{j=1}^\infty\frac{\alpha(1-\alpha)^j}{4}+\sum_{j=-1}^{-\infty}\frac{\beta(1-\beta)^{-j}}{4}
=\frac13.$$\\

\section*{ Acknowledgements}

The work supported by the fundamental project (number: F-FA-2021-425)  of The Ministry of Innovative Development of the Republic of Uzbekistan.

\end{document}